\documentclass[12pt,a4paper]{article}

\usepackage{amsmath,amsthm,amssymb,amscd,a4wide}

\usepackage{graphicx}
\usepackage{caption}
\usepackage{subcaption}

\usepackage{dsfont}
\usepackage{mathrsfs}
\usepackage{setspace}
\usepackage{hyperref}

\numberwithin{equation}{section}

\newtheorem{theorem}{Theorem}[section]

\newtheorem{cor}[theorem]{Corollary}
\newtheorem{remark}[theorem]{Remark}

\theoremstyle{definition}

\numberwithin{equation}{section}

\begin{document}

\thispagestyle{empty}

\vspace*{1cm}

\begin{center}

{\LARGE\bf On a two-particle bound system on the half-line} \\

\vspace*{2cm}

{\large Joachim Kerner \footnote{E-mail address: {\tt Joachim.Kerner@fernuni-hagen.de}} and Tobias M\"{u}hlenbruch \footnote{E-mail address: {\tt Tobias.Muehlenbruch@fernuni-hagen.de}} }%

\vspace*{5mm}

Department of Mathematics and Computer Science\\
FernUniversit\"{a}t in Hagen\\
58084 Hagen\\
Germany\\

\end{center}

\vfill

\begin{abstract} In this paper we provide an extension of the model discussed in \cite{KM16} describing two singularly interacting particles on the half-line $\mathbb{R}_+$. In this model, the particles are interacting only whenever at least one particle is situated at the origin. Stimulated by \cite{QUnruh} we then provide a generalisation of this model in order to include additional interactions between the particles leading to a molecular-like state. We give a precise mathematical formulation of the Hamiltonian of the system and perform spectral analysis. In particular, we are interested in the effect of the singular two-particle interactions onto the molecule. 
\end{abstract}

\newpage

\section{Introduction}

Singular many-particle interactions on general compact quantum graphs were introduced in \cite{BKSingular,BKContact} in order to provide a model for the investigation of many-particle quantum chaos. Based on that, a model of two singularly interacting particles on the half-line $\mathbb{R}_+=[0,\infty)$ (a simple non-compact quantum graph) was formulated in \cite{KM16}. More precisely, the Hamiltonian of this model is formally given by 
\begin{equation}\label{FormalHamiltonian}
H=-\frac{\partial^2}{\partial x^2}-\frac{\partial^2}{\partial y^2}+v(x,y)\left[\delta(x)+\delta(y)\right]\ ,
\end{equation}
$v:\mathbb{R}^2 \rightarrow \mathbb{R}$ being some symmetric (real-valued) interaction potential. Due to the properties of the $\delta$-potential we see from \eqref{FormalHamiltonian} that the two particles are interacting only whenever at least one of the particles is situated at the origin. Furthermore, given the support of $v$ is contained in $B_{\varepsilon}(0)$, i.e., the open ball of radius $\varepsilon > 0$ around $0 \in \mathbb{R}^2$, the two particles are interacting only whenever one particle is situated at the origin and the other is $\varepsilon$-close to it. 

From a physical point of view, the interesting property of \eqref{FormalHamiltonian} is that the two-particle interactions are spatially localised onto the origin. Usually, one expects the two-particle interaction to depend on the relative coordinate only. However, as outlined in \cite{KM16}, there are situations where a spatial localisation of many-particle interactions is expected due to certain inhomogeneities of the real physical system to be modeled. In particular, the authors refer to so called composite wires from the field of applied superconductivity \cite{FossheimSuperconducting}. Those consist of superconducting parts as well as normal-conducting parts and due to the Cooper pairing-effect of superconductivity, the interaction between a pair of electrons in such a wire depends on their corresponding spatial position. In this sense, the Hamiltonian \eqref{FormalHamiltonian} might be used to model a system of two electrons in a wire which is normal-conducting except for a relatively small part at the beginning of the wire which is superconducting. In addition, as described in \cite{glasser1993solvable,glasser2005solvable}, non-separable quantum two-body problems are quite rarely discussed in the literature although having important applications in condensed matter physics as well as quantum entanglement.

In this paper we are interested in an extension of the model discussed in \cite{KM16} by adding to \eqref{FormalHamiltonian} an attractive binding potential between the particles leading to a molecular-like state. This was motivated by \cite{QUnruh} (see Eq.~(5) therein) where the scattering of a two-particle bound system at mirrors is investigated. In particular, the authors are interested in the scattering process given each particle of the molecule is scattered separately at a mirror (each of the mirrors is modelled by a $\delta$-potential). Actually, by allowing for a non-constant potential $v$ in \eqref{FormalHamiltonian}, the model discussed in this paper also provides an extension of the model discussed in \cite{QUnruh}. In their language, the scattering of one particle of the molecule at its mirror is then no longer independent from the position of the second particle. From a physical point of view, this seems to be a possible scenario.

The paper is organised as follows: In Section~\ref{Sec1} we formulate the model and rigorously construct the Hamiltonian of the system via a suitable quadratic form. In Section~\ref{Sec2} we perform spectral analysis and describe the essential as well as the discrete part of the spectrum. We prove the existence of an eigenstate below the essential spectrum in the case of vanishing singular interactions and which is due to the geometry of the one-particle configuration space. We then investigate the effect of additional singular two-particle interactions on the spectrum and prove, as a main result, that the discrete part of spectrum becomes trivial given the singular interactions are repulsive and strong enough.

\section{The model}\label{Sec1}

In this paper we consider a system of two (distinguishable) particles moving on the half-line $\mathbb{R}_+=[0,\infty)$ described by the formal Hamiltonian
\begin{equation}\label{FormalHamiltonianNEW}
H_b=-\frac{\partial^2}{\partial x^2}-\frac{\partial^2}{\partial y^2}+v(x,y)\left[\delta(x)+\delta(y)\right]\ + V_b(|x-y|)\ ,
\end{equation}
$V_b:\mathbb{R}_+ \rightarrow (-\infty,\infty]$ being some binding-potential. For simplicity, we choose $V_b$ to be given by
\begin{equation}\label{BindingPotential}
V_b(|x-y|):=
\begin{cases}
0 \quad \text{if} \quad |x-y| \leq d\ , \\
\infty \quad \text{else}\ ,
\end{cases}
\end{equation}
$d > 0$ characterising the ``size'' of the molecule. Due to the presence of the binding-potential, the two-particle configuration space has been reduced from $\mathbb{R}^2_+$ to $\Omega$ which is given by
\begin{equation}
\Omega:=\{(x,y) \in \mathbb{R}^2_+\ | \ |x-y| \leq d\}\ .
\end{equation}
For later purposes we also define 
\begin{equation}
\partial \Omega_{\sigma}:=\{(x,y) \in \Omega\ |\  x=0 \quad \lor \quad y=0 \}\ ,
\end{equation}
and 
\begin{equation}
\partial \Omega_{D}:=\{(x,y) \in \Omega\ |\  |x-y|=d \}\ .
\end{equation}
Now, in order to arrive at a rigorous realisation of \eqref{FormalHamiltonianNEW} we construct a suitable quadratic form on the Hilbert space $L^2(\Omega)$. We define
\begin{equation}\label{QuadraticForm}
q_{d}[\varphi]:=\int_{\Omega}|\nabla \varphi|^2 \ \mathrm{d}x -\int_{\partial \Omega_{\sigma}} \sigma(y)|\varphi_{bv}|^2 \ \mathrm{d}y\ ,
\end{equation} 
on the domain $\mathcal{D}_q:=\{ \varphi \in H^1(\Omega): \varphi|_{\partial \Omega_{D}}=0 \}$. We note that the Dirichlet boundary conditions along $\partial \Omega_{D}$ are induced by \eqref{BindingPotential}. Furthermore, we set $\sigma(y):=-v(0,y)$ and $\varphi_{bv}=:\varphi|_{\partial \Omega_{\sigma}}$ which are well defined according to the trace theorem for Sobolev functions \cite{Dob05}. Note that $\sigma$ is assumed to be real-valued throughout the paper.
\begin{remark}\label{RemarkBoundaryConditions} The quadratic form \eqref{QuadraticForm} corresponds to a variational formulation of a boundary-value problem for the two-dimensional Laplacian $-\Delta$ on $\Omega$ 
	with coordinate dependent Robin boundary conditions along $\partial \Omega_{\sigma}$ and Dirichlet boundary conditions along $\partial \Omega_{D}$. Indeed, the boundary conditions along $\partial \Omega_{\sigma}$ read
	\begin{equation}\label{RBConditions}\begin{split}
	\frac{\partial \varphi}{\partial n}(0,y)+\sigma(y)\varphi(0,y)&=0\ , \quad \text{and} \\
	\frac{\partial \varphi}{\partial n}(y,0)+\sigma(y)\varphi(y,0)&=0\ ,
	\end{split}
	\end{equation}
	for a.e. $y \in [0,d]$. Here $\frac{\partial}{\partial n}$ denotes the inward normal derivative along $\partial \Omega_{\sigma}$.

	\end{remark}

Using the methods from \cite{KM16} we can directly establish the following statement which generalises Theorem 2.2 of \cite{KM16}. 
\begin{theorem} Let $\sigma \in L^{\infty}(0,d)$ with $0 < d \leq \infty$ be given. Then $q_{d}[\cdot]$ is densely defined, closed and semi-bounded from below. 
\end{theorem}
Hence, according to the representation theorem for quadratic forms \cite{BEH08}, there exists a unique self-adjoint operator being associated with $q_{d}[\cdot]$. This operator, being the Hamiltonian of our system, shall be denoted by $-\Delta^d_{\sigma}$ and his domain by $\mathcal{D}(-\Delta^d_{\sigma}) \subset \mathcal{D}_q$. 
\begin{remark}We note that the case $d=\infty$ corresponds to the model where no binding pontential is added in \eqref{FormalHamiltonianNEW}. In other words, the operator $-\Delta^\infty_{\sigma}$ is the self-adjoint Hamiltonian of the model discussed in \cite{KM16}.
\end{remark}

\section{Spectral properties of $-\Delta^d_{\sigma}$}\label{Sec2}
\subsection{On the essential spectrum}
In this section we are interested in the spectral properties of the self-adjoint operator $-\Delta^d_{\sigma}$ for values $0 < d < \infty$. We start by characterising the essential spectrum and obtain the following result.
\begin{theorem}\label{TheoremEssentialSpectrum} Let $\sigma \in L^{\infty}(0,d)$ with $0 < d < \infty$ be given. Then one has
	\begin{equation}
	\sigma_{ess}(-\Delta^d_{\sigma})=[\pi^2 / 2d^2,\infty).
	\end{equation}
\end{theorem}
\begin{proof}We sketch the proof using similar methods as employed in the proof of Theorem~3.1 of \cite{KM16}. 
	
	In order to prove that $[\pi^2 / 2d^2,\infty) \subset \sigma_{ess}(-\Delta^d_{\sigma})$ we consider the rectangle $D_{k_n,l_n} \subset \Omega$, $l_n > k_n$, which is obtained by dissection of $\Omega$ using the two straight lines $y_1=-x+(2k_n-d)$ and $y_2=-x+(2l_n-d)$ (see Fig.~1~(a)). On $D_{k_n,l_n}$ one then defines $\varphi_n$ to be the (normalised) ground state eigenfunction of the two-dimensional Laplacian subjected to Dirichlet boundary conditions. In particular, due to a separation of variables, the corresponding ground state eigenvalue $E_{n}$ is calculated to be
	\begin{equation}
	E_{n}=\frac{\pi^2}{2d^2}+\frac{\pi^2}{2(l_n-k_n)^2}\ .
	\end{equation} 
	Hence, by letting $l_n,k_n \rightarrow \infty$ appropriately as $n \rightarrow \infty$, we see that $(\varphi_n)_{n \in \mathbb{N}}$ is a Weyl sequence for any value $\lambda \in [\pi^2 / 2d^2,\infty)$. 
	
	On the other hand, to show that $\inf \sigma_{ess}(-\Delta^d_{\sigma})=\pi^2 / 2d^2$ one employs a bracketing argument \cite{BEH08}, i.e., one considers the direct sum $-\Delta^N_{\Omega_1} \oplus -\Delta^N_{D_{k,\infty}}$ of two-dimensional Laplacians with $\Omega_1:=\Omega\setminus D_{k,\infty}$ and with the boundary conditions of the original problem along $\partial \Omega$ (see Fig.~1~(b)). We note that the index $N$ refers to additional Neumann boundary conditions along the line obtained along the construction of $D_{k,\infty}$. Since $\sigma_{ess}(-\Delta^N_{\Omega_1})=\emptyset$ ($-\Delta^N_{\Omega_1}$ is a Laplacian on a bounded Lipschitz domain and has purely discrete spectrum) and since $-\Delta^N_{\Omega_1} \oplus -\Delta^N_{D_{k,l}}$ is smaller than $-\Delta^d_{\sigma}$ in terms of operators 
	we conclude that
	\begin{equation}\begin{split}
	\inf \sigma_{ess}(-\Delta^N_{\Omega_1} \oplus -\Delta^N_{D_{k,\infty}})&=\inf \sigma_{ess}(-\Delta^N_{D_{k,\infty}}) \\
	&\leq \inf \sigma_{ess}(-\Delta^d_{\sigma})\ .
	\end{split}
	\end{equation}
	Furthermore, since $\inf \sigma_{ess}(-\Delta^N_{D_{k,\infty}})=\pi^2 / 2d^2$ by using the reasoning from above as well as a separation of variables argument, we arrive at the statement.
\end{proof}

\begin{figure}
	\centering
	\begin{subfigure}[b]{0.3\textwidth}
		\includegraphics[width=\textwidth]{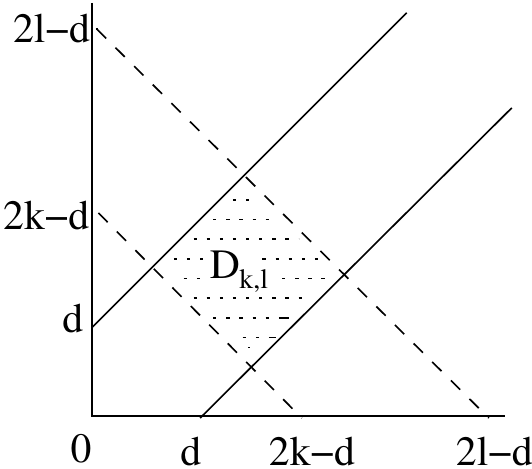}
		\caption{The domain $D_{k,l}$}
		\label{figure1a}
	\end{subfigure}
	\quad 
	\begin{subfigure}[b]{0.3\textwidth}
		\includegraphics[width=\textwidth]{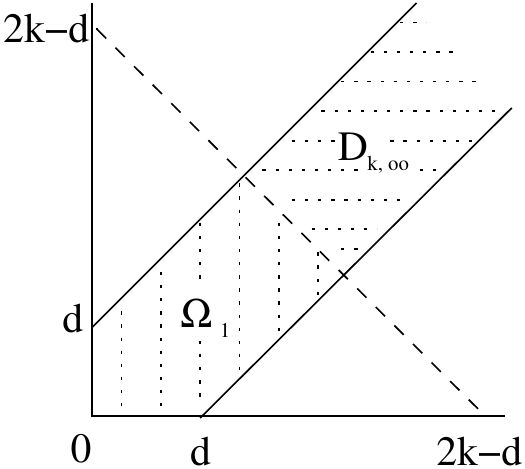}
		\caption{Domains $\Omega_1$ and $D_{k,\infty}$}
		\label{figure1b}
	\end{subfigure}
	\caption{Illustration of the domains $D_{k,l}$, $D_{k,\infty}$ and $\Omega_1$}\label{figure1}
\end{figure}

\begin{remark} By Theorem~\ref{TheoremEssentialSpectrum} we see that $\inf \sigma_{ess}(-\Delta^d_{\sigma}) \rightarrow 0$ as $d \rightarrow \infty$. This is in accordance with Theorem~3.1 of \cite{KM16} given $\sigma \in L^{\infty}(0,\infty)$ fulfils some additional properties, e.g., $|\sigma(y)| \rightarrow 0$ as $y \rightarrow \infty$. However, if one considers the case $\sigma(y):=\sigma > 0$ for a.e. $y\in [0,d]$ or $y \in [0,\infty)$ then $\inf \sigma_{ess}(-\Delta^d_{\sigma})$ does not converge to $\inf \sigma_{ess}(-\Delta^\infty_{\sigma})$, see Remark~3.2 of \cite{KM16}.
\end{remark}
\subsection{On the discrete spectrum}
By Theorem~\ref{TheoremEssentialSpectrum} we know that the bottom of the essential spectrum is located at $\pi^2/2d^2$. One might now ask for the existence of isolated eigenvalues below the essential spectrum. In a first result we investigate the seemingly trivial case of vanishing boundary potential, i.e., $\sigma\equiv 0$. From a physical point of view, this means that the ``molecule'' is moving freely on the half-line with no singular two-particle interactions present. However, quite surprisingly, we prove the existence of an eigenvalue below the bottom of the essential spectrum already in this case.
\begin{theorem}\label{TheoremBoundState} Let $-\Delta^d_{\sigma}$ be given with $\sigma\equiv 0$ and $0 < d < \infty$. Then there exists at least one eigenvalue below the essential spectrum, i.e., 
	\begin{equation}
	\sigma_d(-\Delta^d_{\sigma}) \neq \emptyset\ .
	\end{equation}
\end{theorem}
\begin{proof} We use the results of \cite{ExnerLShaped}. In this paper, the authors prove the existence of an eigenstate at energy $\lambda \approx 0,93\cdot (\pi/b)^2$ of the two-dimensional Laplacian on the L-shaped domain
	\begin{equation}\begin{split}
	\Omega_{b}:=\{(x,y) \in &\mathbb{R}^2_+\ | \ x \in [0,\infty) \ \land \  0 \leq y \leq b \} \\ &\cup \ \{(x,y) \in \mathbb{R}^2_+\ | \ 0 \leq x \leq b \ \land \  y \in [0,\infty) \}  
	\end{split}
	\end{equation}
	subjected to Dirichlet boundary conditions (see Fig.~2~(a)). We note that the parameter $b > 0$ of their model is related to the parameter $d > 0$ of our model via the relation $b=\sqrt{2}d$. We will also denote the corresponding eigenfunction of their model by $\varphi \in H^1(\Omega_b)$.
	
	The strategy now consists of employing a variational argument. For this we consider the two-dimensional Laplacian on an extended version of $\Omega_{b}$, i.e., on
	\begin{equation}\begin{split}
	\Omega_b^{\prime}:=\{(x,y) \in &\mathbb{R}^2\ | \ x \in (-\infty,\infty) \ \land \  0 \leq y \leq b \} \\
	&\ \quad \cup \ \{(x,y) \in \mathbb{R}^2\ | \ 0 \leq x \leq b \ \land \  y \in (-\infty,\infty) \}  \ ,
	\end{split}
	\end{equation}
	again subjected to Dirichlet boundary conditions (see Fig.~2~(b)). One then readily verifies that $\varphi \in H^1(\Omega_b)$, extended by zero onto $\Omega_b^{\prime}$, yields the upper bound $\lambda$ to the ground state energy of the Laplacian defined on $\Omega_b^{\prime}$ (Rayleigh-Ritz variational principle). Let then $\varphi^{\prime} \in H^1(\Omega_b^{\prime})$ denote the corresponding ground state eigenfunction of the two-dimensional Laplacian on $\Omega_b^{\prime}$. The important step is now to dissect $\Omega_b^{\prime}$ into four parts employing the two straight lines $y=x$ and $y=-x+b$ (see Fig.~2~(c)). Each obtained part is a (rotated and translated) copy of our original domain $\Omega$ and hence we can consider the restrictions of $\varphi^{\prime}\neq 0$ to any of the parts. This leaves us with (at most) four trial functions for our original problem. Taking into account that each such restriction vanishes if and only if its gradient does (due to the Dirichlet boundary conditions) we can employ the inequality 
	\begin{equation}\label{EquationSumme}
	\frac{\sum_{i}x_i}{\sum_{i} y_i}  \geq \min_{i}\{x_i/y_i\}\ , \quad x_i,y_i > 0\ ,
	\end{equation}
to obtain the statement via the Rayleigh-Ritz variational principle.
\end{proof} 

\begin{figure}
	\centering
	\begin{subfigure}[b]{0.3\textwidth}
		\includegraphics[width=\textwidth]{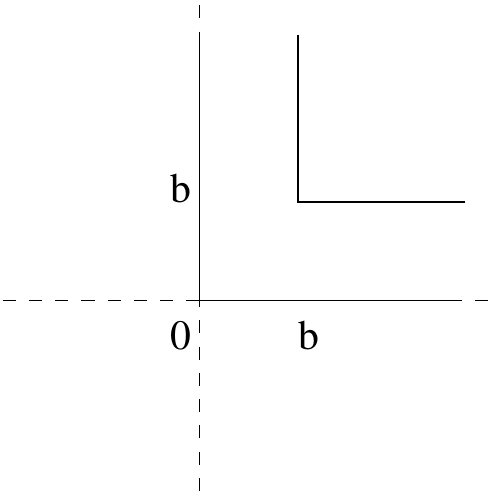}
		\caption{The domain $\Omega_b$}
		\label{figure2a}
	\end{subfigure}
	~ 
	\begin{subfigure}[b]{0.3\textwidth}
		\includegraphics[width=\textwidth]{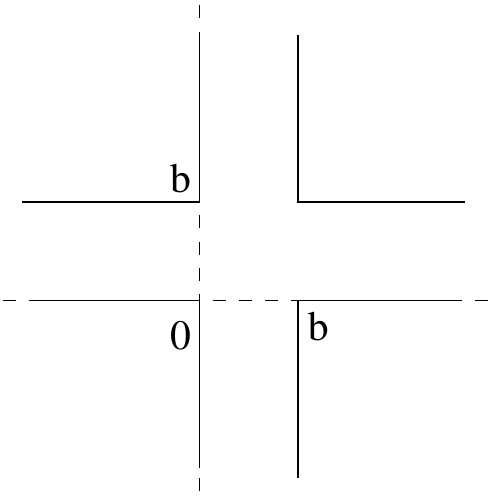}
		\caption{the domain $\Omega_b^\prime$}
		\label{figure2b}
	\end{subfigure}
	~ 
	\begin{subfigure}[b]{0.3\textwidth}
		\includegraphics[width=\textwidth]{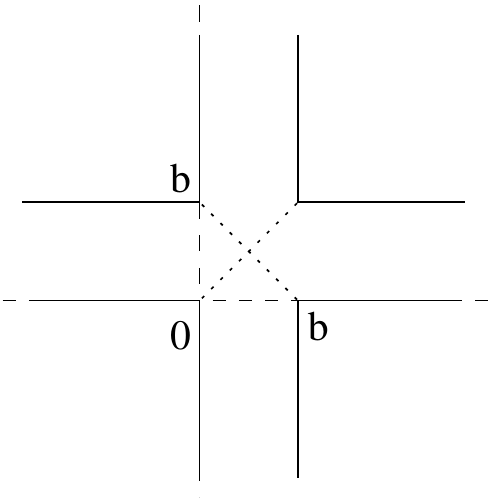}
		\caption{The subdivisions of $\Omega_b^\prime$}
		\label{figure2c}
	\end{subfigure}
	\caption{Illustration of the domains $\Omega_b$, $\Omega_b^\prime$ and the subdivision of $\Omega_b^\prime$ in its 4 parts}\label{figure2}
\end{figure}

\begin{remark} The statement of Theorem~\ref{TheoremBoundState} is, from a physical point of view, truly remarkable. It means that, only due to the geometry of the one-particle configuration space (the half-line), there exists a bound state such that the molecule remains localised around the origin without the presence of attractive external potentials. To compare, if one considers the same problem with the complete real line as one-particle configuration space, no such bound state will exist. 
\end{remark}
We can now investigate the stability of the discrete spectrum in the case where $\sigma \neq 0$, i.e., in the presence of singular two-particle interactions. From a physical point of view, we expect that purely attractive singular interactions do not lead to a destruction of the discrete part of the spectrum. Also, given the repulsive singular interactions are not too strong, the same is expected.
\begin{cor}\label{Corollary} For given $\sigma\in L^{\infty}(0,d)$ and $0  <d < \infty$ the following holds: 
	\begin{enumerate}
		\item[1.] If $\sigma(y) \geq 0$ for almost every $y\in [0,d]$, then $\sigma_d(-\Delta^d_{\sigma}) \neq \emptyset$\ .
		\item[2.] There exists a constant $c > 0$ such that $\sigma_d(-\Delta^d_{\sigma}) \neq \emptyset$ for all $\sigma$ with $\|\sigma\|_{\infty} < c$.
	\end{enumerate}
\end{cor}
\begin{proof} Let $\varphi_0 \in H^1(\Omega)$ denote the (normalised) ground state eigenfunction to the eigenvalue $E_0 > 0$ of $-\Delta^d_{\sigma}$ with $\sigma\equiv 0$.
	
	If $\sigma(y) \geq 0$ for almost every $y\in [0,d]$, we directly obtain 
	\begin{equation}
	q_d[\varphi_0] \leq E_0\ ,
	\end{equation}
	from \eqref{QuadraticForm}. Hence the statement follows from the Rayleigh-Ritz variational principle. 
	
	On the other hand, we have the estimate
	\begin{equation}\begin{split}
	\Bigl|\int_{\partial \Omega_{\sigma}} \sigma(y)|\varphi_{0,bv}|^2 \ \mathrm{d}y\Bigr| &\leq \|\sigma\|_{\infty} \cdot \int_{\partial \Omega_{\sigma}}  |\varphi_{0,bv}|^2 \ \mathrm{d}y\\
	&< \frac{\pi^2}{2d^2}-E_0\ 
	\end{split}
	\end{equation}
	given that $\|\sigma\|_{\infty}$ is small enough. Consequently, one concludes that $q_d[\varphi_0] < \pi^2/2d^2$ for such $\sigma$ and hence the second statement follows. 
\end{proof}
Corollary~\ref{Corollary} shows that the discrete spectrum is stable with respect to small singular two-particle interactions. However, due to the singular nature of the two-particle interactions one might wonder whether $\sigma_d(-\Delta^d_{\sigma}) \neq \emptyset$ holds for an arbitrary boundary potential $\sigma \in L^{\infty}(0,d)$. The following statement shows that this is not the case.
\begin{theorem}\label{TheoremDiscreteSpectrumRepulsive} Let $0 <d < \infty $ be given. Then there exists a constant $\gamma < 0$ such that $\sigma_d(-\Delta^d_{\sigma})=\emptyset$ for all $\sigma \in L^{\infty}(0,d)$ for which $\sigma(y) \leq \gamma$ for almost every $y \in [0,d]$. 
\end{theorem}
\begin{proof} We first note that, due to the assumptions, $q_{\gamma}[\cdot] \leq q_{\sigma}[\cdot]$ where $q_{\gamma}[\cdot]$ denotes the corresponding quadratic form for which $\sigma(y)=\gamma$, $y \in [0,d]$.
	
	 Assume there exists a (normalised) eigenfunction $\varphi \in H^1(\Omega)$ of $-\Delta^d_{\sigma}$ to an eigenvalue $E_0 < \pi^2/2d^2$. Then $-\Delta^d_{\gamma}$, being the self-adjoint operator associated with $q_{\gamma}[\cdot]$, also has a lowest eigenvalue $\tilde{E}_0$ smaller than $\pi^2/2d^2$. Let $\tilde{\varphi} \in H^1(\Omega)$ denote the corresponding normalised eigenfunction. We can then restrict it to the triangle spanned by the points $(0,0),(0,d),(d,0) \in \mathbb{R}^2_+$ and reflect across the line $y=-x+d$ to obtain a trial function for the two-dimensional Laplacian on a square of side length $d$ subjected to (repulsive) Robin boundary conditions (see Fig.~(3) and Remark~\ref{RemarkBoundaryConditions} and note that repulsive refers to $\sigma$ being negative). Note that the restriction of $\tilde{\varphi}$ onto the triangle, as well as the restriction of its gradient, cannot vanish: If the  gradient vanished then $\tilde{\varphi}$ would be constant on the triangle and, due to the Robin boundary conditions along $\partial \Omega_{\sigma}$ (see Remark~\ref{RemarkBoundaryConditions}), this would imply that $\tilde{\varphi}=0$ on the triangle. On the other hand, if $\tilde{\varphi}=0$ on the triangle, then $\tilde{\varphi} \in H^1(\Omega)$ would be a trial function for the two-dimensional Laplacian on $D_{d,\infty}$ (see proof of Theorem~\ref{TheoremEssentialSpectrum}) with Dirichlet boundary conditions along $\partial \Omega_{D}$ and Neumann boundary conditions elsewhere. However, the spectrum of this Laplacian starts at $\pi^2/2d^2$ as can be seen by a separation of variables and we therefore obtain a contradiction. 
	 
	 As described above, using $\tilde{\varphi}\in H^1(\Omega)$ we constructed a trial function on the square of side length $d$ for the Laplacian subjected to (repulsive) Robin boundary conditions. Let   $E^{S}_0 > 0$ denote the lowest eigenvalue of this Laplacian. Then, using \eqref{EquationSumme}, we conclude that 
	 \begin{equation}\label{EqProofXXX}
	 0 < E^{S}_0 < E_0 < \pi^2/2d^2\ .
	 \end{equation}
	 
\begin{figure}
	\centering
	\includegraphics[width=0.3\textwidth]{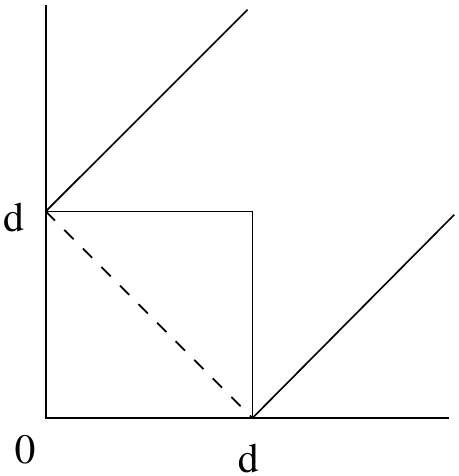}
	
	\caption{The rectangle domain used in the proof of Theorem~\ref{TheoremDiscreteSpectrumRepulsive}}\label{figure3}
\end{figure}
	On the other hand, employing a separation of variables we know that $E^{S}_0=2\lambda_0(\gamma)$ where $\lambda_0(\gamma) > 0$ is the lowest eigenvalue of the one-dimensional Laplacian on an interval of length $d$, subjected to (repulsive) Robin boundary conditions with constant $|\gamma|$. In this case, however, one can show (see Eq.~(4.4) of \cite{BolEnd09}) that $\lambda_0(\gamma) \rightarrow \pi^2/d^2$ as $\gamma \rightarrow -\infty$. Hence, for $|\gamma| > 0$ large enough, we are in contradiction with inequality \eqref{EqProofXXX} and the statement follows.
\end{proof}
\begin{remark} Theorem~\ref{TheoremDiscreteSpectrumRepulsive} is interesting from a physical point of view. It shows that the additional two-particle interactions, albeit their singular nature, destabilise the system by leading to the destruction of the discrete part of the spectrum.
\end{remark}
%


\section*{Acknowledgements}
JK would like to thank S. Egger for helpful discussions. We would also like to thank M.~L.~Glasser for his interest in our model and for providing stimulating references.

{\small
	\bibliographystyle{amsalpha}
	\bibliography{Literature}}

\end{document}